\documentclass[preprint, aip, jmp]{revtex4-1}
\usepackage{amssymb,amsthm,amsmath}
\usepackage{enumerate}
\usepackage{natbib}
\usepackage{mathrsfs} 

\newtheorem{thm}{Theorem}
\newtheorem*{thm*}{Theorem}
\newtheorem{prop}{Proposition}
\newtheorem{lemma}{Lemma}
\newtheorem{cor}{Corollary}

\providecommand{\norm}[1]{\lVert#1\rVert}
\providecommand{\abs}[1]{\lvert#1\rvert}
\providecommand{\inner}[2]{\langle#1,#2\rangle}

\begin{document}

\title{The classical limit of a state on the Weyl algebra}

\author{Benjamin H.~Feintzeig
}
\affiliation{Department of Philosophy\\University of Washington}
\date{\today}

\begin{abstract}
This paper considers states on the Weyl algebra of the canonical commutation relations over the phase space $\mathbb{R}^{2n}$.  We show that a state is regular iff its classical limit is a countably additive Borel probability measure on $\mathbb{R}^{2n}$.  It follows that one can ``reduce'' the state space of the Weyl algebra by altering the collection of quantum mechanical observables so that all states are ones whose classical limit is physical.
\end{abstract}

\maketitle

\section{Introduction}
\label{sec:intro}
In quantum theories we often restrict attention to \emph{regular} states, thereby ruling out non-regular states as unphysical or pathological.  Yet a number of authors have investigated the possibility of using nonregular states as well for some purposes in quantum mechanics,\citep{BeMaPeSi74,FaVeWe74,CaMoSt99, Ha04} quantum field theory,\citep{AcMoSt93} and quantum gravity.\citep{CoVuZa07,As09}  The restriction to regular states is often seen as merely instrumental or for technical convenience, and so one might be tempted to follow these authors and pursue the use of nonregular states.  However, the goal of this paper is to propose a principled justification for focusing on regular states.  We will show that the regular states play a privileged role in explaining the success of classical physics.  

We explain the success of classical physics through the \emph{classical limit}.  The classical limit of a quantum state, in the sense discussed in this paper, is the result of looking at approximations on larger and larger scales until the effects of quantum mechanics disappear.  The central result of this paper shows that we can recover all of the classical states of a classical theory by taking the classical limits of only regular states.  The classical limits of non-regular states play no role in classical physics.

We will deal only with theories with finitely many degrees of freedom modeled by a phase space $\mathbb{R}^{2n}$, in which case the physical states of the classical theory are countably additive Borel probability measures on $\mathbb{R}^{2n}$.  Quantum states are then positive, normalized linear functionals on the C*-algebra of (bounded) physical magnitudes, or \emph{observables}, satisfying the canonical commutation relations.  This algebra is known as the Weyl algebra over $\mathbb{R}^{2n}$.  The central purpose of this paper is thus to establish the following claim:

\begin{thm*}
A state on the Weyl algebra over $\mathbb{R}^{2n}$ is regular iff its classical limit is a countably additive Borel probability measure on $\mathbb{R}^{2n}$.
\end{thm*}

\noindent The intended interpretation is that one need only use the classical limits of regular quantum states to explain the success of classical physics.

Moreover, this discussion has import for the construction of quantum theories.  We construct quantum theories, at least in the algebraic approach, by constructing their algebra of observables.  But in many cases there is no widespread consensus about precisely which algebra we should use to represent the observables of a given quantum system.\citep{AsIs92,Em97,Fe17a}  Although the Weyl algebra over $\mathbb{R}^{2n}$ is one commonly used tool, it always allows for non-regular states.  If one has the desire to restrict attention to regular states, then one might look for another algebra that allows for only regular states.\footnote{For more on alternatives to the Weyl algebra for a C*-algebra of the canonical commutation relations, see \citet{We84,La90b,BuGr08,GrNe09}}  We demonstrate that such an algebra exists.

It is known that each C*-algebra uniquely determines its state space.\citep{AlSh01}  And if this state space is ``too large'' in the sense that it contains unphysical or pathological states, then there is a general procedure (established in previous work) for constructing a new C*-algebra whose state space consists in precisely the collection of physical states:\footnote{For related work, see also \citet{La95}.}
\begin{thm}[\citet{Fe17a}]
Let $\mathfrak{A}$ be a C*-algebra and let $V\subseteq\mathfrak{A}^*$.  Then there exists a unique C*-algebra $\mathfrak{B}$ and a surjective *-homomorphism $f:\mathfrak{A}\rightarrow\mathfrak{B}$ such that $\mathfrak{B}^*\cong V$ with the isomorphism given by $\omega\in\mathfrak{B}^*\mapsto (\omega\circ f)\in V$ iff the following conditions hold:
\begin{enumerate}[(i)]
\label{thm:reduction}
\item V is a weak* closed subspace; and
\item for all $A,B\in\mathfrak{A}$,
\[\sup_{\omega\in V;\norm{\omega}=1}\abs{\omega(AB)}\leq \sup_{\omega\in V;\norm{\omega}=1}\abs{\omega(A)}\sup_{\omega\in V;\norm{\omega}=1}\abs{\omega(B)}\]
\end{enumerate}
\end{thm}

\noindent To prove Thm. \ref{thm:reduction}, one shows the C*-algebra $\mathfrak{B}$ is *-isomorphic to $\mathfrak{A}/\mathcal{N}(V)$, where $\mathcal{N}(V)$ is the ideal given by the annihilator of $V$ in $\mathfrak{B}$.  This allows one to choose a subspace of $\mathfrak{A}^*$ satisfying (i) and (ii) and construct a unique C*-algebra with precisely the desired dual space.  Thus, we can use Thm. \ref{thm:reduction} to find a new algebra---constructed from the Weyl algebra---which allows for only regular states.  The quantum states on this new algebra suffice for explaining the success of classical physics through the classical limit.

The result of this process---a C*-algebra with so-called \emph{full regularity}---is already well known,\citep{La90b,GrNe09} but it is usually obtained by means of twisted group algebras. \citet{GrNe09} attempt to take this standard approach and extend it to infinite-dimensional phase spaces, where finite-dimensional tools fail.  The goal of the current paper, by contrast, is to establish a new approach already in the finite-dimensional case, with the hope that future work on this approach can illuminate constructions in infinite dimensions.  The additional benefit of the new approach, however, will be a \emph{justification} for the regularity condition by appeal to the classical limit, demonstrating the importance of future work in this area.

The structure of the paper is as follows.  In \S\ref{sec:prelim}, I define the Weyl algebra over $\mathbb{R}^{2n}$ and the notion of a regular state.  In \S\ref{sec:limits}, I clarify the notion of the ``classical limit" of a state on the Weyl algebra using a continuous field of C*-algebras.  In Section \ref{sec:states}, I establish some small lemmas to characterize the countably additive Borel measures on $\mathbb{R}^{2n}$ in terms of algebraic structure, and I apply Thm. \ref{thm:reduction} to the purely classical system with phase space $\mathbb{R}^{2n}$.  Finally, in Section \ref{sec:result}, I prove the main result and discuss its significance.  The results of this paper involve little, if any, mathematical novelty.  I hope, however, that the perspective I provide on the construction of new quantum theories is of interest.

\section{Preliminaries}
\label{sec:prelim}

The Weyl algebra over $\mathbb{R}^{2n}$ is formed by deforming the product of the C*-algebra $AP(\mathbb{R}^{2n})$ of complex-valued almost periodic functions on $\mathbb{R}^{2n}$.  The C*-algebra $AP(\mathbb{R}^{2n})$ is generated by functions $W_0(x):\mathbb{R}^{2n}\rightarrow\mathbb{C}$ for each $x\in\mathbb{R}^{2n}$ defined by
\[W_0(x)(y) := e^{i x\cdot y}\]
for all $y\in\mathbb{R}^{2n}$, where $\cdot$ is the standard inner product on $\mathbb{R}^{2n}$.  Polynomials (with respect to pointwise multiplication, addition, and complex conjugation) of functions of the form $W_0(x)$ for $x\in\mathbb{R}^{2n}$ are norm dense in $AP(\mathbb{R}^{2n})$ with respect to the standard supremum norm.\footnote{For more on the algebra of almost periodic functions, see \citet{He53,BiHoRi04a}.}

The \emph{Weyl algebra} over $\mathbb{R}^{2n}$, denoted $\mathcal{W}_h(\mathbb{R}^{2n})$, for $h\in(0,1]$ is generated from the same set of functions by defining a new multiplication operation. The symbol $W_h(x)\in\mathcal{W}_h(\mathbb{R}^{2n})$ is used now to denote the element $W_0(x)$ as it is considered in the new C*-algebra.  Define the non-commutative multiplication operation on $\mathcal{W}_h(\mathbb{R}^{2n})$ by
\[W_h(x)W_h(y) := e^{\frac{ih}{2}\sigma(x,y)}W_h(x+y)\]
for all $x,y\in\mathbb{R}^{2n}$, where $\sigma$ is the standard symplectic form on $\mathbb{R}^{2n}$ given by
\[\sigma((a,b),(a',b')) := a'\cdot b - a\cdot b'\]
for $a,b,a',b'\in\mathbb{R}^n$ and $\cdot$ is now the standard inner product on $\mathbb{R}^n$.  The Weyl algebra is the norm completion in the minimal regular norm\citep{MaSiTeVe74,Pe90,BiHoRi04a} of polynomials of elements of the form $W_h(x)$ for $x\in\mathbb{R}^{2n}$ with respect to the non-commutative multiplication operation.

It is this C*-algebra (sometimes known as the \emph{CCR algebra}, or the \emph{Weyl form of the CCRs}) that is often used to model the physical magnitudes, or \emph{observables}, of a quantum mechanical system constructed from a classical system with phase space $\mathbb{R}^{2n}$.  One can then take the positive, normalized linear functionals, or \emph{states}, on the C*-algebra of physical magnitudes to model the physically realizable states of the quantum system.\citep{Ha92,BrRo87,BrRo96,Ho97}

However, it is known that many of the states on $\mathcal{W}_h(\mathbb{R}^{2n})$ are pathological, in the sense that they violate continuity conditions some believe to be necessary for physics. \citep{Ar95,Ha01,Ha04,Ru11a}  Among the states one might consider pathological are the \emph{non-regular} states.\footnote{The Weyl algebra does indeed allow for many nonregular states; see \citet{BeMaPeSi74,FaVeWe74,AcMoSt93}.  However, one can still gain some control of the representations of nonregular states.\citep{CaMoSt99}}  A bounded linear functional $\omega$ on $\mathcal{W}_h(\mathbb{R}^{2n})$ is called \emph{regular} just in case for all $x\in\mathbb{R}^{2n}$, the mapping
\[t\in\mathbb{R}\mapsto\omega(W_h(tx))\]
is continuous.  The well known Stone-von Neumann theorem \citep{Pe90, Su99} tells us that if (and only if) a state is regular, its GNS representation is quasiequivalent to the ordinary Schr\"{o}dinger representation on the Hilbert space $L^2(\mathbb{R}^n)$, and hence leads to a reconstruction of the orthodox formalism for quantum mechanics.  In this case, one sees immediately that the restriction to regular states is desirable for applications in physics.

When one moves, however, to the context of quantum field theory, the story so far breaks down.  The Stone-von Neumann theorem is no longer applicable to the phase space of a field-theoretic system, which is infinite-dimensional, and hence, fails to be locally compact.  Instead, infinitely many (even regular) unitarily inequivalent irreducible representations of the Weyl algebra appear, and one has no principled way of using any particular irreducible representation to motivate a restriction to some subset of states.

The goal of this paper, however, is to provide a principled method for restricting the state space of the Weyl algebra by looking at its purely algebraic structure, rather than its Hilbert space representations.  In particular, by understanding the algebra $\mathcal{W}_h(\mathbb{R}^{2n})$ as part of a \emph{strict and continuous deformation quantization} of $AP(\mathbb{R}^{2n})$, one can use purely algebraic tools to show that the regular states are already privileged.  With this quantization, for any state $\omega$ on $\mathcal{W}_h(\mathbb{R}^{2n})$, one can construct the continuous field of states that form a ``constant" section of linear functionals.  We will call the value of this section at $h = 0$ the \emph{classical limit} of the state $\omega$.  The privileged states on $\mathcal{W}_h(\mathbb{R}^{2n})$ for our purposes will be the ones whose classical limits can be used to explain the success of classical physics.  We will understand the classical limit of a quantum state to be useful for explaining the success of classical physics just in case it is a physical state on $AP(\mathbb{R}^{2n})$.

A natural candidate for the collection of physical states of a classical theory with phase space $\mathbb{R}^{2n}$ is the collection of countably additive Borel probability measures on $\mathbb{R}^{2n}$.  But these states form a proper subset of the state space of $AP(\mathbb{R}^{2n})$.  Hence we will think of the countably additive Borel probability measures on $\mathbb{R}^{2n}$ as privileged classical states for our purposes.  To explain the success of classical physics, we only need to find enough quantum states to recover the privileged classical states in the classical limit.  Our main result shows that if one accepts that the countably additive Borel probability measures are the only physical states of the classical theory with phase space $\mathbb{R}^{2n}$, then only the regular states on $\mathcal{W}_h(\mathbb{R}^{2n})$ are needed to explain the success of classical physics through the classical limit.

\section{Classical limits}
\label{sec:limits}

In this section, we make precise the notion of the \emph{classical limit} of a quantum state.  Define the ``classical" Weyl algebra as $\mathcal{W}_0(\mathbb{R}^{2n}):= AP(\mathbb{R}^{2n})$.  Consider the family of C*-algebras $\{\mathcal{W}_h(\mathbb{R}^{2n})\}_{h\in[0,1]}$ with ``quantization" maps, denoted $\mathcal{Q}_h:\mathcal{W}_0(\mathbb{R}^{2n})\rightarrow\mathcal{W}_h(\mathbb{R}^{2n})$, given by the unique linear norm continuous extension of
\[\mathcal{Q}_h(W_0(x)) := W_h(x)\]
for all $x\in\mathbb{R}^{2n}$ and each $h\in[0,1]$.  For each $h\in[0,1]$, the map $\mathcal{Q}_h$ has a norm dense range in $\mathcal{W}_h(\mathbb{R}^{2n})$ and satisfies
\[\mathcal{Q}_h(A^*) = \mathcal{Q}_h(A)^*\]
for all $A\in\mathcal{W}_0(\mathbb{R}^{2n})$.  Furthermore, as can be seen in the realization of $\mathcal{W}_h(\mathbb{R}^{2n})$ as Toeplitz operators, the map $\mathcal{Q}_h$ is positive and norm continuous.\citep{BeCo86}  One should be careful \emph{not} to confuse this map $\mathcal{Q}_h$ with other maps sometimes referred to as ``Weyl Quantization'' (e.g., in \citet{La98b,DuHeSm00}) that fail to be positive because they have a different domain.  Similarly, one should note that the map $\mathcal{Q}_h$ we consider is norm continuous on the domain and codomain considered here even though the map $x\mapsto W_h(x)$ is \emph{not} norm continuous.

Given this quantization map, it is known that:
\begin{enumerate}[(i)]
\item For all ``suitably smooth"\citep{BiHoRi04b}  $A,B\in\mathcal{W}_0(\mathbb{R}^{2n})$,
\[\lim_{h\rightarrow 0}\norm{\frac{i}{h}[\mathcal{Q}_h(A),\mathcal{Q}_h(B)] - \mathcal{Q}_h(\{A,B\})} = 0\]
where $\{\cdot,\cdot\}$ is the standard Poisson bracket on $\mathbb{R}^{2n}$ corresponding to the symplectic form $\sigma$ and $[\cdot,\cdot]$ is the commutator (defined by $[X,Y]:= XY-YX$).
\item For all $A,B\in\mathcal{W}_0(\mathbb{R}^{2n})$,
\[\lim_{h\rightarrow 0}\norm{\mathcal{Q}_h(A)\mathcal{Q}_h(B)-\mathcal{Q}_h(AB)} = 0\]
\item For all $A\in\mathcal{W}_0(\mathbb{R}^{2n})$, the mapping $h\mapsto\norm{\mathcal{Q}_h(A)}$ is continuous. 
\end{enumerate}
Furthermore, for each $h\in[0,1]$, $\mathcal{Q}_h$ is injective and $\mathcal{Q}_h[\mathcal{W}_0(\mathbb{R}^{2n})]$ is closed under the product in $\mathcal{W}_h(\mathbb{R}^{2n})$.  This structure $(\mathcal{W}_h,\mathcal{Q}_h)_{h\in[0,1]}$ is thus a \emph{strict deformation quantization}.\citep{Co92,La98b,La06,La17,Ri94,BiHoRi04b}$^,$\footnote{These strict quantizations are meant to represent non-perturbative or analytic formulations of the classical limit of quantum theories, as opposed to the \emph{formal deformation quantizations} investigated by, e.g., \citet{Wa15}, \citet{Ko03}.}

In fact, this structure gives rise to a \emph{continuous bundle of C*-algebras} $(\{\mathcal{W}_h\}_{h\in[0,1]},\mathcal{K})$ over the base space $[0,1]$.\citep{Di77,La98b,La17}  The C*-algebra of continuous sections $\mathcal{K}$ is given by the unique C*-subalgebra of $\prod_{h\in[0,1]}\mathcal{W}_h(\mathbb{R}^{2n})$ containing the elements of the form
\[h\mapsto\mathcal{Q}_h(A)\]
for all $A\in\mathcal{W}_0(\mathbb{R}^{2n})$.  Defining the global quantization map $\mathcal{Q}:\mathcal{W}_0(\mathbb{R}^{2n})\rightarrow\mathcal{K}$ by
\[\mathcal{Q}(A) := [h\mapsto\mathcal{Q}_h(A)]\]
for all $A\in\mathcal{W}_0(\mathbb{R}^{2n})$, we have that $(\{\mathcal{W}_h(\mathbb{R}^{2n})\}_{h\in[0,1]},\mathcal{K};\mathcal{Q})$ is a \emph{continuous quantization}.\citep{La98b,La17,BiHoRi04b}$^,$\footnote{For the case of a possibly infinite dimensional pre-symplectic phase space, see \citet{BiHoRi04b,HoRi05}.  For extensions to larger algebras see \citet{HoRiSc08}.}

Now we can consider \emph{continuous fields of linear functionals}\citep{La98b} as families $\{\omega_h\}_{h\in[0,1]}$, where $\omega_h\in\mathcal{W}_h(\mathbb{R}^{2n})^*$ for each $h\in[0,1]$,  and for each continuous section $\varphi\in\mathcal{K}$,
\[\lim_{h\rightarrow 0}\omega_h(\varphi(h)) = \omega_0(\varphi(0))\]
In particular, for a fixed $H\in[0,1]$ and $\omega\in\mathcal{W}_H(\mathbb{R}^{2n})^*$, since $\mathcal{Q}_H$ is norm continuous we can define $\hat{\omega}\in\mathcal{W}_0(\mathbb{R}^{2n})^*$ by
\[\hat{\omega} := \omega\circ\mathcal{Q}_H\]
And we then construct the ``constant" continuous field of functionals $\{\omega_h\}_{h\in[0,1]}$ through $\omega_0 = \hat{\omega}$ and $\omega_H = \omega$ by defining $\omega_h:\mathcal{W}_h(\mathbb{R}^{2n})\rightarrow\mathbb{C}$ as the unique norm continuous extension of
\[\omega_h(\mathcal{Q}_h(A)) := \hat{\omega}(A)\]
for all $A\in\mathcal{W}_0(\mathbb{R}^{2n})$.  It is easy to see that $\{\omega_h\}_{h\in[0,1]}$ is indeed a continuous field of functionals, so we will define $\hat{\omega}$ as the \emph{classical limit} of $\omega$.  In particular, since for any $h\in[0,1]$, the quantization map $\mathcal{Q}_h$ is positive, we know that when $\omega\in\mathcal{W}_h(\mathbb{R}^{2n})^*$ is a state, its classical limit $\hat{\omega}$ is a state as well.\footnote{For more on the relationship between the classical and quantum state spaces, see \citet[][Thm. 13]{BeCo86}.}  We will use this notion of the classical limit in what follows.

Before proceeding, a caveat is in order:  the $\hbar\rightarrow 0$ limit described in this section is of course only \emph{one} way, among many others, to make precise the classical limit of quantum theories.  Furthermore, when one uses the $\hbar\rightarrow 0$ limit as outlined above, the classical limit of a state is understood only \emph{relative} to the continuous quantization considered.  Nevertheless, we choose to investigate classical limits defined relative to the above continuous quantization.  The current paper demonstrates this notion of the classical limit is fruitful at least for understanding the significance of regular and non-regular states.

\section{Classical states}
\label{sec:states}

In this section, we look in more detail at the structure of the state space of a classical theory.  This will help us understand the import of Thm. \ref{thm:reduction} for C*-algebras of physical magnitudes, and it will lead to some lemmas required for the proof of the main result.  We focus on a classical theory whose phase space is given by a locally compact, $\sigma$-compact Hausdorff topological space $\mathcal{M}$.  Although in the next section we restrict attention to finite-dimensional symplectic vector spaces, for the moment we make no such restriction.

Prima facie, it's not clear which C*-algebra of functions on $\mathcal{M}$ we should take to model the bounded physical magnitudes of the classical system.  Candidates include:
\begin{itemize}
\item $B_\Sigma(\mathcal{M})$, the algebra of bounded complex-valued functions on $\mathcal{M}$ that are measurable with respect to some $\sigma$-algebra $\Sigma$.
\item $C_b(\mathcal{M})$, the algebra of bounded continuous complex-valued functions on $\mathcal{M}$.
\item $C_0(\mathcal{M})$, the algebra of continuous complex-valued functions on $\mathcal{M}$ vanishing at infinity.
\end{itemize}
Of course, we have the following inclusion relations:
\[C_0(\mathcal{M})\subseteq C_b(\mathcal{M})\subseteq B(\mathcal{M})\]
This list of possible algebras of classical quantities is not exhaustive, but it (along with some additions later on) will suffice for our purposes.  See \citet{We84} for more on alternative classical algebras and their corresponding quantum algebras.

We will have occasion to examine measurable functions with respect to two particular $\sigma$-algebras.  First, let $\Sigma_B(\mathcal{M})$ be the $\sigma$-algebra of Borel sets generated by the topology on $\mathcal{M}$.  We will denote the bounded Borel measurable functions by $B_B(\mathcal{M}) :=B_{\Sigma_B}(\mathcal{M})$.  Second, let $\Sigma_{R}(\mathcal{M})$ be the $\sigma$-algebra of universally Radon measurable sets.\citep{Fr03}  We will denote the bounded measurable functions with respect to $\Sigma_{R}$ by $B_{R}(\mathcal{M}):=B_{\Sigma_{R}}(\mathcal{M})$.  Notice also that since every Borel set is universally Radon measurable, we have $\Sigma_B(\mathcal{M})\subseteq \Sigma_{R}(\mathcal{M})$, so it follows that $B_B(\mathcal{M})\subseteq B_{R}(\mathcal{M})$.

The states on $C_0(\mathcal{M})$ are precisely the countably additive regular Borel probability measures on $\mathcal{M}$.\citep{FeDo88,La98b,La17}  So if we think the physical classical states consist in the countably additive regular Borel probability measures, then this gives us reason to use $C_0(\mathcal{M})$ to model the physical observables of the classical system.  If we choose a larger C*-algebra, then we will in general allow for more states.  However, we know that the bidual of $C_0(\mathcal{M})$ is $B_{R}(\mathcal{M})$.\footnote{See Prop. 437I of \citet{Fr03}.  For an alternative characterization of the bidual of $C_0(\mathcal{M})$, see Thm. 5.12 on p. 76 of \citet{Co90}.}  Further, it follows that every element of $B_{R}(\mathcal{M})$ can be approximated weakly by nets of elements of $C_0(\mathcal{M})$, and the normal state space of $B_{R}(\mathcal{M})$, understood as a W*-algebra is also the collection of countably additive regular Borel probability measures on $\mathcal{M}$.  So we can also use $B_{R}(\mathcal{M})$ to model possibly idealized (even discontinuous) physical magnitudes, where we understand elements of the normal state space of $B_{R}(\mathcal{M})$ to model physical states.\footnote{For more on the physical interpretation of the bidual and weak limits, see \citet{Fe17c}.}  If we desire, we can also restrict attention to $B_B(\mathcal{M})$, which should alter our interpretation only minimally.  For example, we also know that every physical magnitude in $B_B(\mathcal{M})\subseteq B_{R}(\mathcal{M})$ can likewise be approximated weakly by nets of elements of $C_0(\mathcal{M})$.

Notice, however, that in the quantization procedure of \S\ref{sec:limits} we focused on a different algebra, the algebra of almost periodic functions.\footnote{For continuous quantizations that start from the classical algebra $C_0(\mathcal{M})$, see, e.g., the discussion of \emph{Berezin quantization} in \citet{La98b}.}  When $\mathcal{M}$ is a locally compact abelian group, we know that $AP(\mathcal{M})\subseteq C_b(\mathcal{M})$, but $AP(\mathcal{M})$ and $C_0(\mathcal{M})$ are in general not identical.  As such, the state space of $AP(\mathcal{M})$ will in general not be equal to that of $C_0(\mathcal{M})$.\citep{He53,Ru62}  In this section, we will characterize the countably additive Borel measures as a subspace of the bounded linear functionals on $AP(\mathcal{M})$.

In fact, we will work with a more general collection of algebras of functions.  For the remainder of this section, let $\mathfrak{C}\subseteq C_b(\mathcal{M})$ be an abelian C*-subalgebra of the bounded continuous complex-valued functions on $\mathcal{M}$ satisfying the constraints: (i) $\mathfrak{C}$ separates points of $\mathcal{M}$; and (ii) $\mathfrak{C}$ contains the constants.  So $\mathfrak{C}$ may be the almost periodic functions, or it may be another algebra altogether.  We characterize the bounded countably additive regular Borel measures as a subset $V$ of $\mathfrak{C}^*$, and we show that applying Thm. \ref{thm:reduction} with this choice of $V$ transforms $\mathfrak{C}$ into $B(\mathcal{M})$.

To perform this procedure, we need to notice that $\mathfrak{C}$ is both ``too small" and ``too large".  It is ``too small" in the sense that it does not contain discontinuous functions like projections, which are needed for the spectral theory that allows us to interpret physical magnitudes.  It is ``too large" in the sense that it allows for ``states at infinity", which may be considered pathological for physical applications.  Thus, we will enlarge the algebra $\mathfrak{C}$ by completing it in its weak topology to form $\mathfrak{C}^{**}$.  Then we will reduce the algebra by restricting ourselves to a privileged collection of states on $\mathfrak{C}^{**}$.

Let us denote the pure state space of $\mathfrak{C}$ by $\mathscr{P}:= \mathscr{P}(\mathfrak{C})$.  Recall that $\mathfrak{C}$ is *-isomorphic to $C(\mathscr{P})$, the continuous functions on the pure state space of $\mathfrak{C}$ with the weak* topology, which is a compact Hausdorff space.  Furthermore, there is a continuous injection $k: \mathcal{M}\rightarrow \mathscr{P}$, called a \emph{compactification},\footnote{See \citet[(p. 16)][]{Ga69} and \citet[(\S4.4)][]{KaRi97}.  For the case of particular interest for this paper where $\mathfrak{C} = AP(\mathbb{R}^{2n})$, the compactification $\mathscr{P}(AP(\mathbb{R}^{2n}))$ is known as the \emph{Bohr compactification} of $\mathbb{R}^{2n}$; see \citet{AnKa43,He53,Ru62}.} such that for any $f\in\mathfrak{C}$ and its surrogate $\hat{f}\in C(\mathscr{P})$, $\hat{f}\circ k = f$.  The map $k$ sends each point in $\mathcal{M}$ to the pure state it determines; $k$ is defined for all $p\in\mathcal{M}$ by 
\begin{align*}
k(p)(f) := f(p)
\end{align*}
for all $f\in\mathfrak{C}$.  Similarly, all states (even mixed states) can be represented as Borel measures on $\mathscr{P}$.  The Riesz-Markov theorem \citep{ReSi80,Ru87,Fr03,La17} implies that each $\omega\in\mathfrak{C}^*$ corresponds to a unique bounded regular Borel measure $\mu_\omega$ such that\footnote{See also \citet{La91} and \citet{Fe16a} for similar discussions in the context of classical physics.}
\[\omega(f) = \int_{\mathscr{P}}\hat{f}~d\mu_\omega\]
Here, the Borel $\sigma$-algebra on $\mathscr{P}$ is determined by the weak* topology on $\mathscr{P}$.  The weak completion of $\mathfrak{C}$ is then $B_{R}(\mathscr{P})\cong \mathfrak{C}^{**}$, the bounded functions on the pure state space of $\mathfrak{C}$ that are measurable with respect to the $\sigma$-algebra of universally Radon measurable sets.

First, we need to establish that $\mathfrak{C}^{**}$ is indeed ``large enough" to contain all of the discontinuous functions on $\mathcal{M}$ we desire.

\begin{lemma}
\label{lemma: Borel set}
If $S\in\Sigma_B(\mathcal{M})$, then $k[S]\in\Sigma_B(\mathscr{P})$, where $\mathscr{P}$ is treated with the weak* topology.
\end{lemma}

\begin{proof}
Since the Borel sets are generated by the closed sets, it suffices to prove the claim for the case where $S$ is a closed set.  So suppose $S\subseteq\mathcal{M}$ is a closed set.  Since $\mathcal{M}$ is $\sigma$-compact, there are countably many compact subsets $K_n\subseteq \mathcal{M}$ such that $S = \bigcup_nK_n$.  Now since $k$ is continuous, for each $n\in\mathbb{N}$, we know $k[K_n]$ is compact and hence closed because $\mathscr{P}$ is Hausdorff.  Thus $k[S] = \bigcup_nk[K_n]$ is a Borel set.
\end{proof}

\begin{lemma}
\label{lemma:R set}
If $S\in\Sigma_R(\mathcal{M})$ then $k[S]\in\Sigma_R(\mathscr{P})$, where $\mathscr{P}$ is treated with the weak* topology.
\end{lemma}

\begin{proof}
By Lemma \ref{lemma: Borel set}, $k$ is a measurable isomorphism between the measurable spaces $(\mathcal{M},\Sigma_B(\mathcal{M}))$ and $(k[\mathcal{M}],\Sigma_B(k[\mathcal{M}])$, and hence $k$ must also be a measurable isomorphism between the measurable spaces $(\mathcal{M},\Sigma_{R}(\mathcal{M}))$ and $(k[\mathcal{M}],\Sigma_{R}(k[\mathcal{M}]))$.  Thus, if $S\in\Sigma_{R}(\mathcal{M})$, we must have $k[S]\in\Sigma_{R}(k[\mathcal{M}])$.  Now it suffices to show that when $k[S]\subseteq k[\mathcal{M}]$ is universally Radon measurable in $k[\mathcal{M}]$, it is also universally Radon measurable in $\mathscr{P}$, i.e. if $k[S]\in\Sigma_{R}(k[\mathcal{M}])$, then $k[S]\in\Sigma_{R}(\mathscr{P})$.

So suppose $k[S]\in\Sigma_{R}(k[\mathcal{M}])$.  Then suppose $\mu$ is a Radon measure on $\mathscr{P}$.  The restriction $\mu_{|k[\mathcal{M}]}$ is a Radon measure on $k[\mathcal{M}]$ and hence $k[S]$ is $\mu_{|k[\mathcal{M}]}$-measurable.  It follows that there is a $\mu$-measurable set $U\subseteq\mathscr{P}$ such that $k[S] = U\cap k[\mathcal{M}]$.  Now, since $k[\mathcal{M}]$ is Borel and hence $\mu$-measurable, it follows that $k[S]$ is $\mu$-measurable.  Therefore, $k[S]\in\Sigma_{R}(\mathscr{P})$.
\end{proof}

\begin{lemma}
\label{lemma:function}
Suppose $f:\mathscr{P}\rightarrow\mathbb{C}$ has support contained in $k[\mathcal{M}]$.  Then

\begin{enumerate}
\item $f\in B_B(\mathscr{P})$ iff $f\circ k\in B_B(\mathcal{M})$.

\item $f\in B_{R}(\mathscr{P})$ iff $f\circ k\in B_{R}(\mathcal{M})$.
\end{enumerate}
\end{lemma}

\begin{proof}
\begin{enumerate}
\item ($\Rightarrow$) Suppose $f$ is Borel measurable on $\mathscr{P}$.  Then clearly $f\circ k$ is Borel measurable because $k$ is continuous.

($\Leftarrow$) Suppose $f\circ k$ is Borel measurable on $\mathcal{M}$.  Let $S$ be a measurable set in $\mathbb{C}$.  Then $(f\circ k)^{-1}[S]$ is Borel measurable in $\mathcal{M}$.  Let $k[\mathcal{M}]^C$ denote $\mathscr{P}\setminus k[\mathcal{M}]$.  Since $k$ is an injection and $\text{supp }f\subseteq k[\mathcal{M}]$, we know that $f^{-1}[S] = k[(f\circ k)^{-1}[S]]\cup k[\mathcal{M}]^C$ if $0\in S$.  On the other hand, if $0\notin S$, then  $f^{-1}[S] = k[(f\circ k)^{-1}[S]]$.  The set $k[(f\circ k)^{-1}[S]]$ is Borel measurable by Lemma \ref{lemma: Borel set}, and $k[\mathcal{M}]^C$ is Borel measurable because $k[\mathcal{M}]$ is measurable by Lemma \ref{lemma: Borel set}.  Hence $f^{-1}[S]$ is Borel measurable.

\item ($\Rightarrow$) Suppose $f\in B_{R}(\mathscr{P})$.  Let $S$ be a measurable set in $\mathbb{C}$.  Then $f^{-1}[S]\in \Sigma_{R}(\mathscr{P})$, so $f^{-1}[S]\cap k[\mathcal{M}]\in\Sigma_{R}(k[\mathcal{M}])$ since $k[\mathcal{M}]\in\Sigma_B(\mathscr{P})\subseteq\Sigma_{R}(\mathscr{P})$.  It follows that $(f\circ k)^{-1}[S]\in \Sigma_{R}(\mathcal{M})$ since $k$ is an isomorphism between the measurable spaces $(\mathcal{M},\Sigma_{R}(\mathcal{M}))$ and $(k[\mathcal{M}],\Sigma_{R}(k[\mathcal{M}]))$.  Hence $f\circ k\in B_{R}(\mathcal{M})$.

($\Leftarrow$)  Suppose $f\circ k\in B_{R}(\mathcal{M})$.  Again, let $S$ be a Borel set in $\mathbb{C}$.  Then we can use the same reasoning as in the proof of (1) but applying Lemma \ref{lemma:R set} in place of Lemma \ref{lemma: Borel set}.  We find that Lemma \ref{lemma:R set} yields $f^{-1}[S]\in \Sigma_{R}(\mathscr{P})$.
\end{enumerate}
\end{proof}

Now we know both that $B_B(\mathscr{P})$ contains surrogates for all elements of $B_B(\mathcal{M})$ and that $\mathfrak{C}^{**}\cong B_{R}(\mathscr{P})$ contains surrogates for all elements of $B_{R}(\mathcal{M})$.  Hence, we will see that these algebras are ``large enough" that we can apply Thm. \ref{thm:reduction} to them and in doing so obtain the algebras $B_B(\mathcal{M})$ and $B_{R}(\mathcal{M})$, respectively, by restricting attention to an appropriate state space.  We define
\[V^C_0 := \{\omega\in\mathfrak{C}^*~|~\text{$\mu_\omega(\mathscr{P}\setminus k[\mathcal{M}]) = 0$}\}\]
where $\mu_\omega$ is the measure on $\mathscr{P}$ determined by $\omega$ through the Riesz-Markov theorem.  Furthermore, we define $V^C$ to be the weak* closure of $V^C_0$ in $\mathfrak{C}^{***}$.  It follows that $V^C$ satisfies conditions (i) and (ii) of Thm. \ref{thm:reduction}.  We show that $V^C$ is precisely the collection of functionals that reduce the algebra $\mathfrak{C}^{**}\cong B_{R}(\mathscr{P})$ to $B_{R}(\mathcal{M})$.

\begin{prop}
\label{prop:classreduction}
The C*-algebra $\mathfrak{B} = \mathfrak{C}^{**}/\mathcal{N}(V^C)$ of Thm. \ref{thm:reduction} is *-isomorphic to $B_{R}(\mathcal{M})$.
\end{prop}

\begin{proof}
Notice that for any functions $f,g\in \mathfrak{C}^{**}$,
\begin{align*}
f+\mathcal{N}(V^C) = g + \mathcal{N}(V^C) && \Longleftrightarrow && f(x) = g(x) \text{ for all $x\in k[\mathcal{M}]$}
\end{align*}
Hence, the map
\[f + \mathcal{N}(V^C)\mapsto f\circ k\]
is an injective *-homomorphism of $\mathfrak{C}^{**}/\mathcal{N}(V^C)$ into $B_{R}(\mathcal{M})$.  Furthermore, this mapping is surjective because for each $h\in B_{R}(\mathcal{M})$, we have $h = f\circ k$ for the function $f:\mathscr{P}\rightarrow\mathbb{C}$ defined by
\begin{align*}
f(x) = \begin{cases}
h(k^{-1}(x)) & \text{ if $x\in k[\mathcal{M}]$}\\
0 & \text{ otherwise}
\end{cases}
\end{align*}
and moreover this function $f$ is in $B_{R}(\mathscr{P})$ by (2) of Lemma \ref{lemma:function}.  It follows that the mapping $f+\mathcal{N}(V^C)\mapsto f\circ k$ provides the desired *-isomorphism from $\mathfrak{C}^{**}/\mathcal{N}(V^C)$ to $B_{R}(\mathcal{M})$.
\end{proof}

\begin{cor}
The C*-algebra $\mathfrak{B} = B_{B}(\mathscr{P})/\mathcal{N}(V^C)$ of Thm. \ref{thm:reduction} is *-isomorphic to $B_B(\mathcal{M})$.
\end{cor}

\begin{proof}
This follows immediately from (1) of Lemma \ref{lemma:function}.
\end{proof}

\noindent Thus,
we can use $V^C$ to reduce the state space of an algebra of continuous bounded functions on $\mathcal{M}$ and in doing so recover an algebra of bounded measurable functions on $\mathcal{M}$.  In particular, when we use the algebra $\mathfrak{C} = AP(\mathbb{R}^{2n})$, we can reduce its state space and recover an algebra of bounded measurable functions on $\mathbb{R}^{2n}$.

The vector space $V_0^C$ has an intuitive physical significance.  We now show that $V_0^C$ can indeed be characterized as the collection of bounded countably additive regular Borel measures on $\mathcal{M}$, which we will use in what follows.  Notice that each bounded countably additive regular measure $\mu$ on $\mathcal{M}$ defines a bounded linear functional $\omega_\mu\in\mathfrak{C}^*$ by
\[\omega_\mu(f):=\int_{\mathcal{M}} f~d\mu\]
Let $M(\mathcal{M})$ be the collection of such linear functionals on $\mathfrak{C}$ determined by a bounded countably additive regular Borel measure.
\begin{prop}
\label{prop: ptws cts}
Let $\omega\in\mathfrak{C}^*$.  Then $\omega\in V_0^C$ iff $\omega\in M(\mathcal{M})$.
\end{prop}

\begin{proof}
This follows immediately from the fact that $k$ is a measurable isomorphism between the measurable spaces $(\mathcal{M},\Sigma_B(\mathcal{M}))$ and $(k[\mathcal{M}],\Sigma_B(k[\mathcal{M}]))$, as shown in Lemma \ref{lemma: Borel set}.
\end{proof}

\noindent This establishes that the elements of $V_0^C$ are precisely the bounded countably additive regular Borel measures on $\mathcal{M}$.  In particular, when $\mathfrak{C} = AP(\mathbb{R}^{2n})$, this shows that $V_0^C$ is the collection of bounded countably additive regular Borel measures on $\mathbb{R}^{2n}$.  We will use this fact to prove our main result.

\section{Main result}
\label{sec:result}

Now we restrict attention to the finite-dimensional phase space $\mathcal{M} = \mathbb{R}^{2n}$.  We know from previous results that if we consider the Weyl algebra over $\mathbb{R}^{2n}$ and reduce its state space to the regular states by applying Thm. \ref{thm:reduction}, then we are left with the algebra $\mathscr{B}(\mathscr{H})$ of all bounded linear operators on a separable Hilbert space $\mathscr{H}$.  More precisely, we define
\[V_0^Q := \{\omega\in\mathcal{W}_h(\mathbb{R}^{2n})^*|\text{ } \omega\text{ is regular}\}\]
Further, we define $V^Q$ to be the weak* closure of $V_0^Q$ in $\mathcal{W}_h(\mathbb{R}^{2n})^{***}$.  It follows that $V^Q$ satisfies conditions (i) and (ii) of Thm. \ref{thm:reduction}.  To proceed, we first need to enlarge the algebra $\mathcal{W}_h(\mathbb{R}^{2n})$ to $\mathcal{W}_h(\mathbb{R}^{2n})^{**}$ so that it is ``large enough'', as we did for $\mathfrak{C}$ in \S\ref{sec:states}.  Then we reduce its state space to yield our desired algebra $\mathscr{B}(\mathscr{H})$.
\begin{prop}[\citet{Fe17a}]
\label{prop:weylreduction}
The C*-algebra $\mathfrak{B} = \mathcal{W}_h(\mathbb{R}^{2n})^{**}/\mathcal{N}(V^Q)$ (for $h\neq 0$) of Thm. \ref{thm:reduction} is *-isomorphic to $\mathscr{B}(\mathscr{H})$, where $\mathscr{H}$ is a separable Hilbert space.
\end{prop}
\noindent This should be unsurprising because the well known Stone-von Neumann theorem already shows that the regular states form the folium of the Schr\"{o}dinger representation, which is irreducible.  But notice that there is a striking resemblance between the situation in Props. \ref{prop:classreduction} and \ref{prop:weylreduction}.  In both cases, we alter the algebra of observables by choosing an appropriate state space---$V^C$ or $V^Q$---and applying Thm. \ref{thm:reduction}.  We will clarify the relationship between these propositions by characterizing the relationship between $V_0^Q$, the collection of regular functionals on the Weyl algebra, and $V_0^C$,  the collection of countably additive Borel measures on $\mathbb{R}^{2n}$.  To do so, we now prove the main result.

\begin{thm}
\label{thm:classlim}
Let $\omega\in\mathcal{W}_h(\mathbb{R}^{2n})^*$ (for $h\neq 0$). Then $\omega$ is regular iff $\omega\circ\mathcal{Q}_h$ is a countably additive Borel measure on $\mathbb{R}^{2n}$ (i.e., $\omega\in V_0^Q$ iff $\omega\circ\mathcal{Q}_h\in V_0^C$).
\end{thm}

\begin{proof}

Some preliminaries.  Notice that since the Stone-von Neumann theorem implies the regular functionals are isomorphic to the predual of $\mathscr{B}(L^2(\mathbb{R}^n))$, we know from Thm. 7.4.7 of \citet[(p. 485)][]{KaRi97} that each regular bounded linear functional can be decomposed into a linear combination of regular states.  Hence, it suffices to prove the claim for the case where $\omega$ is a state, which we will assume in what follows.

$(\Rightarrow)$ Suppose $\omega$ is a regular state.  Consider the GNS representation $(\pi_{\omega\circ\mathcal{Q}_h},\mathscr{H}_{\omega\circ\mathcal{Q}_h})$ for the state $\omega\circ\mathcal{Q}_h$ on $\mathcal{W}_0(\mathbb{R}^{2n})$.  We denote the pure state space on $\mathcal{W}_0(\mathbb{R}^{2n})$ by $\mathscr{P}:=\mathscr{P}(\mathcal{W}_0(\mathbb{R}^{2n}))$ and recall that $\mathcal{W}_0(\mathbb{R}^{2n})\cong C(\mathscr{P})$.  Thm. 5.2.6 of \citet[(p. 315)][]{KaRi97} implies there is a projection-valued measure $E$ on $\mathscr{P}$ with the following property.  For any vector $\varphi\in\mathscr{H}_{\omega\circ\mathcal{Q}_h}$, define the Borel measure $\mu_\varphi$ on $\mathscr{P}$ by
\[\mu_{\varphi}(S) := \inner{\varphi}{E(S)\varphi}\]
for any Borel set $S\subseteq \mathscr{P}$.  Then we have
\[\inner{\varphi}{\pi_{\omega\circ\mathcal{Q}_h}(W_0(x))\varphi} = \int_{\mathscr{P}} \hat{W}_0(x)~d\mu_{\varphi}\]
for any $x\in\mathbb{R}^{2n}$ (where $\hat{W}_0(x)$ is here the surrogate of the function $W_0(x)$ on the space $\mathscr{P}$ of pure states).  In particular, for the choice $\varphi = \Omega_{\omega\circ\mathcal{Q}}$, the cyclic vector representing the state $\omega\circ\mathcal{Q}$ from the GNS construction, it follows that $\mu_{\varphi}$ is the Riesz-Markov measure $\mu_{\omega\circ\mathcal{Q}_h}$ on $\mathscr{P}$ corresponding to $\omega\circ\mathcal{Q}_h$.

Now, since $\omega$ is regular, the mapping
\[x\mapsto\pi_{\omega\circ\mathcal{Q}_h}(W_0(x))\]
for all $x\in\mathbb{R}^{2n}$ is a weak operator continuous unitary representation of the topological group $\mathbb{R}^{2n}$.  Thus, the SNAG theorem (\citet[][p. 243]{BrRo87}) implies that we have a projection-valued measure on $\mathbb{R}^{2n}$:
\[S\mapsto E(k[S])\]
for any Borel set $S\subseteq\mathbb{R}^{2n}$ (Recall $\mathbb{R}^{2n}$ is self-dual as a topological group). Here, the map $k:\mathbb{R}^{2n}\rightarrow\mathscr{P}$ is the compactification associated to $\mathcal{W}_0(\mathbb{R}^{2n})$.  The justification is as follows: for any $\varphi\in\mathscr{H}_{\omega\circ\mathcal{Q}_h}$, define the Borel measure $\hat{\mu}_\varphi$ on $\mathbb{R}^{2n}$ by
\[\hat{\mu}_{\varphi}(S) := \mu_{\varphi}(k[S])\]
for any Borel set $S\subseteq\mathbb{R}^{2n}$.  We know $\hat{\mu}_\varphi$ satisfies
\[\inner{\varphi}{\pi_{\omega\circ{Q}_h}(W_0(x))\varphi} = \int_{\mathbb{R}^{2n}}W_0(x)~d\hat{\mu}_\varphi\]
for any $x\in\mathbb{R}^{2n}$ (where $W_0(x)$ is now considered as a function on $\mathbb{R}^{2n}$).

Hence, since $E\circ k$ is a projection valued measure, it must be the case that 
\[\mu_\varphi(k[\mathbb{R}^{2n}]) = \hat{\mu}_\varphi(\mathbb{R}^{2n}) = 1\]
and thus
\[\mu_\varphi(\mathscr{P}\setminus k[\mathbb{R}^{2n}]) = 0\]
This implies that for the choice $\varphi = \Omega_{\omega\circ\mathcal{Q}_h}$,
\[\mu_{\omega\circ\mathcal{Q}_h}(\mathscr{P}\setminus k[\mathbb{R}^{2n}]) = 0\]
Prop. \ref{prop: ptws cts} now implies $\omega\circ\mathcal{Q}_h$ is a countably additive regular Borel measure on $\mathbb{R}^{2n}$.

$(\Leftarrow)$ Suppose $\omega\circ\mathcal{Q}_h$ is a countably additive regular Borel probability measure on $\mathbb{R}^{2n}$.  By Prop. \ref{prop: ptws cts}, we know that the Riesz-Markov measure $\mu_{\omega\circ\mathcal{Q}_h}$ corresponding to $\omega\circ\mathcal{Q}_h$ satisfies
\[\mu_{\omega\circ\mathcal{Q}_h}(\mathscr{P}\setminus k[\mathbb{R}^{2n}]) = 0\]
Now we know that for any $x\in\mathbb{R}^{2n}$,
\begin{equation*}
\begin{split}
\omega\circ\mathcal{Q}_h(W_0(tx)) &= \int_{\mathscr{P}}W_0(tx)~d\mu_{\omega\circ\mathcal{Q}_h}\\
&=\int_{k[\mathbb{R}^{2n}]}W_0(tx)~d\mu_{\omega\circ\mathcal{Q}_h}
\end{split}
\end{equation*}
for all $t\in\mathbb{R}$.  Since the functions $W_0(tx)$ are uniformly bounded by 1 on the domain $k[\mathbb{R}^{2n}]$, the dominated convergence theorem (\citet[][p. 17, Thm. I.11]{ReSi80}) implies that
\[t\mapsto\omega\circ\mathcal{Q}_h(W_0(tx))\]
is a continuous function, which shows that $\omega$ is regular.
\end{proof}

Recall that the state $\omega$ determines a constant continuous field of states $\{\omega_h\}_{h\in[0,1]}$ on the deformation quantization discussed in \S\ref{sec:limits} with classical limit $\hat{\omega} = \omega\circ\mathcal{Q}_h$.  Thus, Thm. \ref{thm:classlim} establishes the central claim of this paper: a state $\omega$ on $\mathcal{W}_h(\mathbb{R}^{2n})$ is regular iff its classical limit $\omega\circ\mathcal{Q}_h$ is a countably additive Borel probability measure on $\mathbb{R}^{2n}$.

\section{Discussion}

Thm. \ref{thm:classlim} shows that we only need the regular quantum states to explain the success of classical physics.  As such, this result bears on the question of which quantum states we ought to take to be physically significant in quantum theories.  For example, some have been interested in constructing quantum theories using non-regular states and their corresponding representations.\citep{Ha04,CoVuZa07}  But these non-regular states are not needed for the explanation of the success of classical physics.  I do not claim that this rules out approaches to physics involving non-regular states, but it does demonstrate some of their counterintuitive features.

Furthermore, this result shows that one can use the classical limit of quantum theories, or more specifically quantum states, to guide the construction of quantum theories.  Prop. \ref{prop:weylreduction} shows that we can transform the Weyl algebra to an algebra whose state space consists precisely of the states whose classical limits are manifestly physically significant.  This may be desirable if one believes that the only quantum states that are physical are the ones whose classical limit is physical.  The result of this process is precisely the orthodox formulation of non-relativistic quantum mechanics for systems whose phase space is $\mathbb{R}^{2n}$.

The main result of this paper thus suggests a possible methodology for choosing an algebra of observables to use in the construction of new quantum theories. As one will notice from a survey of the literature, there are many different approaches to constructing algebras of quantum observables.  One needs to make a choice about what \emph{type} of algebra to use, e.g., *-algebras or C*-algebras (See, e.g., \citet{Re16}).  And even once one has made this choice, one needs to choose which algebra of a given type is appropriate for modeling a given physical system (See, e.g., \citet{AsIs92}).  The methodology the result of this paper suggests is that we look for an algebra that has an appropriate state space---in the sense that the classical limits of the allowed quantum states are physical states in the classical theory.  I do not claim that this methodology is guaranteed to work in the quantization of all classical theories, but merely that it works in the simplest case.  Since this approach is motivated by the desired explanations of the success of previous theories, I suggest that it might be fruitful to apply the same methodology in the quantization of other classical theories.

I hope that the result of this procedure provides some further understanding of quantization and the classical limit, and that this procedure can be extended to illuminate the quantization of other classical theories.  In this vein, I'd like to outline a number of further questions for investigation.

\begin{enumerate}

\item Do the algebras $\mathcal{W}_h(\mathbb{R}^{2n})^{**}/\mathcal{N}(V^Q)\cong \mathscr{B}(\mathscr{H})$ with the natural composition of quotient maps with quantization maps provide a strict or continuous quantization\footnote{Here we would of course require a slight technical alteration of the conditions given in \citet{La98b} to allow for the weakly continuous extension of the quantization from a Poisson subalgebra rather than merely the norm continuous extension. See \citet{Fe17b}.} of the algebra $\mathcal{W}_0(\mathbb{R}^{2n})^{**}/\mathcal{N}(V^C)\cong B_{R}(\mathbb{R}^{2n})$?  If so, is this quantization equivalent to the weakly continuous extension of the Berezin quantization of $C_0(\mathbb{R}^{2n})\subseteq B_{R}(\mathbb{R}^{2n})$ to the compact operators $\mathcal{K}(\mathscr{H})\subseteq\mathscr{B}(\mathscr{H})$?

\item Can one extend the results of \S \ref{sec:states} for field systems whose phase space is not locally compact?  In particular, should we understand the countably additive regular Borel probability measures on a non-locally compact phase space as modeling the physically significant states?  Should we understand algebras of bounded measurable functions on a non-locally compact phase space as modeling physically significant observables?\footnote{See \citet{GrNe09} for an attempt to extend results on regularity to infinite dimensions.  In this vein, one would also hope to make connections with much of the recent work on algebraic quantum field theory.\citep{BrDuFr09,BrFr09,FrRe15,DuFr01,Re16}}

\item Can one apply the methodology of \S \ref{sec:result} to recover the quantization of a system whose phase space is not simply connected?  There are known procedures for arriving at the quantization of a theory whose phase space has the form $G/H$ for some locally compact abelian group $G$ and closed subgroup $H$.\citep{La90,La93a,La95}  Is it possible to arrive at the same quantum theory by quantizing a system with phase space $G$ and restricting attention to only states with the appropriate classical limit?

\end{enumerate}

\noindent I hope that answers to these (or similar) questions might aid both our understanding of the classical limit and the development of tools for constructing strict and continuous deformation quantizations for further systems of physical interest.

\section*{Acknowledgments}

\noindent I would like to thank Thomas Barrett, Samuel Fletcher, Adam Koberinski, James Weatherall, and an anonymous referee for helpful comments and discussions.

\bibliography{algebraicmethods.bib}
\bibliographystyle{mla}

\end{document}